\documentclass{llncs}

\usepackage[utf8]{inputenc}
\usepackage{enumerate}
\usepackage{graphicx}
\usepackage{epstopdf}
\usepackage{amsfonts}
\usepackage{amssymb}
\usepackage{amsmath}
\usepackage{sectsty}
\usepackage{wasysym}
\usepackage{hyperref}
\usepackage{mathtools}
\usepackage{enumitem}

\newcommand{\Sel}{\mathrm{Sel}}
\newcommand{\gen}{\mathrm{gen}}
\newcommand*\ub[1]{\overline{#1}}
\newcommand*\lb[1]{\underline{#1}}

\begin{document}

\title{On convexity and solution concepts in cooperative interval games}

\author{Jan Bok}

\institute{Computer Science Institute, Charles University, Malostransk\'{e} n\'{a}m\v{e}st\'{i} 25, 11800, Prague, Czech Republic, \texttt{bok@iuuk.mff.cuni.cz}}

\maketitle

\begin{abstract}
Cooperative interval game is a cooperative game in which every coalition
gets assigned some closed real interval. This models uncertainty about
how much the members of a coalition get for cooperating together.

In this paper we study convexity, core and the Shapley value of games
with interval uncertainty.

Our motivation to do so is twofold.
First, we want to capture which properties are preserved when we generalize
concepts from classical cooperative game theory to interval games. Second,
since these generalizations can be done in different ways, mainly
with regard to the resulting level of uncertainty, we try to compare
them and show their relation to each other.
\end{abstract}

\section{Introduction}

Uncertainty and inaccurate data are an everyday issue in real-world
situations. Therefore it is important to be able to make decisions even
when the exact data are not available and only bounds on them are known.

In classical cooperative game theory, every group of players
(\emph{coalition}) knows the precise reward for their cooperation; in cooperative
interval games, only the worst and the best possible outcome is known. Such
situations can be naturally modeled with intervals encapsulating these
outcomes.

Cooperative games under interval uncertainty were first considered by Branzei,
Dimitrov and Tijs in 2003 to study bankruptcy situations \cite{bank} and later
further extensively studied by Alparslan G\"{o}k in her PhD thesis
\cite{gokphd} and in follow-up papers (see the
references section of \cite{coop1} ).

We note that there are several other models incorporating
a different level of uncertainty, namely fuzzy cooperative games,
multichoice games, crisp games (see \cite{branzei2008models} for more), or games
under bubbly uncertainty \cite{palanci2014cooperative}.

There are several reasons why it is interesting to study cooperative interval game.
From the aforementioned models of cooperative games, it is a quite simple model but it
is easier to analyze and it is suitable for situations where we do not have any other
assumptions on data we get. There are already a few applications of this model, namely
on forest situations \cite{forest}, airport problems \cite{gok2009airport},
bankruptcy \cite{bank} or network design \cite{d2016network}.

We continue in the line of research started in \cite{BH15} (there is also an updated
version on arXiv \cite{1410.3877}). We focus on selections, that is on possible outcomes
of interval games.

\paragraph{Our results.}
Here is a summary of our results and also how our paper is organized.
\begin{itemize}[topsep=2pt]
	\item Section \ref{sec:con} is about convexity in interval games. We
	characterize selection convex interval games in a style of Shapley \cite{shapley1971cores}.
	\item Section \ref{sec:coincidence} investigates a problem of core coincidence, i.e.\ when 
	the two different versions of generalized core for interval games coincide.
	We partially solve this problem.
	\item Section \ref{sec:shapley} is about the Shapley value for interval games.
	We present a different axiomatization of the interval Shapley extension
	than the one by \cite{cinani} and also show some important properties of this function.
\end{itemize}

\section{Preliminaries}
\label{sec:prelim}

\subsection{Classical cooperative games}

Comprehensive sources on classical cooperative game theory are for example
\cite{branzei2008models,driessen1988cooperative,gilles2010cooperative,peleg2007introduction}.
For more on applications, see e.g. \cite{bilbao2000cooperative,combopt,insurance}.
Here we present only the necessary background
theory for studying interval games. We examine only the games with transferable utility
(TU) and therefore by a cooperative game or a game we mean a cooperative TU game.

\begin{definition}
  \emph{(Cooperative game)} 
  A cooperative game is an ordered pair $(N, v)$, where $N = \{1,2,\ldots ,n\}$
  is a set of players and $v: 2^N \to \mathbb{R}$ is a characteristic function
  of the cooperative game. We further assume that $v(\emptyset) = 0$.
\end{definition}

The set of all cooperative games with a player set $N$ is denoted by $G^N$.
Subsets of $N$ are called \emph{coalitions} and $N$ itself is called the
\emph{grand coalition}. We often write $v$ instead of $(N,v)$ because we can
identify a game with its characteristic function.

\paragraph{Solution concepts.} To further analyze players' gains, we need a \emph{payoff vector} which
can be interpreted as a proposed distribution of rewards between players.

\begin{definition}
  \emph{(Payoff vector)} A payoff vector for a cooperative game $(N, v)$ is a vector
  $x \in \mathbb{R}^N$ with $x_i$ being a reward given to the $i$th player.
\end{definition}

\begin{definition}
  \emph{(Imputation)} An imputation of $(N,v) \in G^N$ is a vector $x \in \mathbb{R}^N$
  such that $\sum_{i \in N} x_i = v(N)$ and $x_i \ge v(\{i\})$ for every $i \in N$.

  The set of all imputations of a given game $(N,v)$ is denoted by
  $I(v)$.
\end{definition}

\begin{definition}
  \emph{(Core)} The core of $(N,v) \in G^N$ is the set
  $$C( v ) = \Bigl\{ x \in I(v); \; \sum_{ i \in S } x_i \geq v(S), \forall S \subseteq N \Bigr\}.$$
\end{definition}

The last solution concept we will write about is the Shapley value. It was introduced by Lloyd Shapely in 1952 \cite{shapley}.
It has many interesting properties; namely, it is a one-point solution concept, it always
exists and it can be axiomatized by very natural axioms.
We refer to \cite{peleg2007introduction} for a survey of results on the Shapley value.

\begin{theorem}\emph{(Shapley, 1952, \cite{shapley1971cores})} There exists a unique
function $f: G^N \to \mathbb{R}^N$, satisfying the following properties for every
$(N,v) \in G^N$.
\begin{itemize}
	\item (Efficiency) It holds that $\sum_{i \in N}f_i(v) = v(N)$.
	\item (Dummy player) It holds $f_i(v) = 0$ for every $i \in N$, such that for every
	$S\setminus \{i\} \subseteq N$, equality $v(S\cup\{i\}) = v(S)$ holds.
	\item (Symmetry) If for every $S \subseteq N \setminus \{ i,j\}$,
	$$v(S \cup i) - v(S) = v(S \cup j) - v(S)$$
	holds, then $f_i(v) = f_j(v)$.
	\item (Additivity) For every two games $u,v \in G^N$, $f_i(u+v) = f_i(u) + f_i(v)$ holds.
\end{itemize}
\end{theorem}

This unique function is called the \emph{Shapley value} ($\phi$) and it is defined as
$$\phi_i(v) \coloneqq \sum_{S \subseteq N\setminus \{i\}} \frac{|S|!(n-|S|-1)!}{n!} (v(S \cup i) - v(S)).$$

\paragraph{Classes.} There are many important classes of cooperative games. Here we show the most
important ones.

\begin{definition} \emph{(Monotonic game)} A game $(N,v)$ is monotonic if for
every $T \subseteq S \subseteq N$ we have
$$v(T) \le v(S)\textrm{.}$$
\end{definition}

Informally, in monotonic games, bigger coalitions are stronger.

Another important type of game is a \emph{convex game}.

\begin{definition} \emph{(Convex game)} A game $(N,v)$ is convex if its
characteristic function is supermodular. The characteristic function
is supermodular if for every $S \subseteq T \subseteq N$,
$$v(T) + v(S) \le v(S \cup T) + v(S \cap T).$$
\end{definition}

Clearly, supermodularity implies superadditivity. The class of convex
games is maybe the most prominent class, it has many applications
and theoretical properties. We present the most important one for this
paper.

\begin{theorem} (Shapley, 1971, \cite{shapley1971cores}) \label{thm:convex}
  Every convex game has a nonempty core.
\end{theorem}

\subsection{Interval analysis}

\begin{definition}
  \emph{(Interval)} An interval $X$ is a set
  $$X \coloneqq [\underline{X},\overline{X}] =\{x \in \mathbb{R}: \underline{X} \le x \le \overline{X}\}\textrm{.}$$
  with $\underline{X}$ being the lower bound and $\overline{X}$ being the upper bound of the interval. The length of an interval $X$ is defined as $|X| \coloneqq | \ub{X} - \lb{X}|$.
\end{definition}

From now on, by an interval we mean a closed interval. The set of
all real closed intervals is denoted by $\mathbb{IR}$.

The following definition (from \cite{moore2009introduction}) shows how to do basic arithmetics with intervals.

\begin{definition} \label{def:ari} \emph{(Interval arithmetics)}
  For every $X, Y, Z \in \mathbb{IR}$, and $0 \notin Z$, define
\begin{align*}
  X + Y &\coloneqq [\underline{X} + \underline{Y}, \overline{X} + \overline{Y}]\textrm{,}\\
  X \ominus Y &\coloneqq [\underline{X} - \overline{Y}, \overline{X} - \underline{Y}]\textrm{,}\\
  X \cdot Y &\coloneqq [\min S , \max S],\ S = \{\underline{X}\overline{Y}, \overline{X}\underline{Y}, \underline{X}\underline{Y}, \overline{X}\overline{Y}\}\textrm{, and}\\
  X\,/\,Z &\coloneqq [\min S , \max S],\ S = \{\underline{X}/\overline{Z}, \overline{X}/\underline{Z}, \underline{X}/\underline{Z}, \overline{X}/\overline{Z}\}\textrm{.}
\end{align*}
\end{definition}

For our purpose, we need to have a slightly different definition of subtraction.
The aforementioned subtraction operator is known as \emph{Moore's subtraction operator}.

\begin{definition}\emph{(Partial subtraction operator)}
For every $I, J \in \mathbb{IR}$, such that $\lb{I} - \lb{J} \le \ub{I}-\ub{J}$, define
$$I - J \coloneqq [\lb{I} - \lb{J},\, \ub{I}-\ub{J}].$$
\end{definition}

In other words, the length of the subtracted interval has to be
lesser or equal to the length of the interval we subtract from.

\begin{example}
	Take two intervals $[1,4]$ and $[3, 5]$. Then  $[1,4] - [3, 5] = [-2,-1]$.
	Notice, however, that $[3,5] - [1, 4]$ is undefined.
\end{example}

We note that this notation is not common in interval analysis. The minus sign is
used for Moore's subtraction operator there. Also, in our previous paper
\cite{BH15,1410.3877} we used minus sign for Moore's subtraction.

Alparslan G\"{o}k \cite{gokphd} choose to compare intervals in the following
way, using the weakly better
operator. That was inspired by Hinojosa et al. \cite{partially}.

\begin{definition} \emph{(Weakly better operator $\succeq$)}
An interval $I$ is weakly better than interval $J$ ($I \succeq J$) if $\underline{I} \ge
\underline{J}$ and $\overline{I} \ge \overline{J}$. Interval $I$ is better than $J$ ($I \succ J$) if and only if $I
\succeq J$ and $I \not= J$.
\end{definition}

Naturally, we also use $A \prec B$ and $C \preceq D$ for $B \succ A$
and $D \succeq C$, respectively.

\subsection{Cooperative interval games}
\label{coopgamesintro}

Now we review basics of cooperative games with interval uncertainty. The following is
the main definition of this paper.

\begin{definition}
  \emph{(Cooperative interval game)} 
  A cooperative interval game is an ordered pair $(N, w)$, where $N = \{1,2,\ldots ,n\}$
  is a set of players and $w: 2^N \to \mathbb{IR}$ is the characteristic function
  of the cooperative game. We further assume that $w(\emptyset) = [0,0]$.
  The set of all interval cooperative games on a player set $N$ is denoted by $IG^N$.
\end{definition}

\begin{note}
	We often write $w(i)$ instead of $w(\{i\})$ and $w(i,j)$ instead of $w(\{i,j\})$.
\end{note}

\begin{note}
Every cooperative interval game in which its characteristic function maps to degenerate
intervals only can be associated with a classical cooperative game. The converse holds as well.
\end{note}

\begin{definition} \emph{(Border games)}
  For every $(N,w) \in IG^N$, border games $(N,\underline{w}) \in G^N$ (lower border game) and $(N,\overline{w}) \in G^N$ (upper border game) are given by
  $\underline{w}(S) \coloneqq \underline{w(S)}$ and $\overline{w}(S)\coloneqq \overline{w(S)}$ for every $S \in 2^N$.
\end{definition}

\begin{definition} \emph{(Length game)} The length game of $(N,w) \in IG^N$ is the game $(N,|w|) \in G^N$ with
  $$|w|(S)\coloneqq \overline{w}(S) - \underline{w}(S),\ \forall S \in 2^N.$$
\end{definition}

\begin{definition} \emph{(Degenerated game)}
	We call a game $(N,w) \in IG^N$ \emph{degenerated} if its length game is everywhere zero,
	that is, $|w|(S) = 0$ for every $S \in 2^N$. A \emph{non-degenerated game} is a game which
	is not degenerated.
\end{definition}

The basic notion of our approach will be a selection and consequently a selection imputation and a selection core.

\begin{definition} \emph{(Selection)} A game $(N,v) \in G^N$ is a selection of $(N,w)
\in IG^N$ if for every $S \in 2^N$ we have $v(S) \in w(S)$. The set of all selections
of $(N,w)$ is denoted by $\Sel(w)$.
\end{definition}

Note that border games are examples of selections and also of degenerated games.

\paragraph{Solution concepts.} There are many possibilities
how to define imputations and core for interval games. We present the following
two. The first one is based on
selections, the second one on the weakly better operator.

\begin{definition} The set of \emph{interval selection
imputations} (or just selection imputations) of $(N,w) \in IG^N$ is defined as
$$\mathcal{SI}(w) = \bigcup \big\{I(v)\ |\ v \in \Sel(w) \big\}.$$
\end{definition}

\begin{definition} \label{def:sc} The \emph{interval selection core} (or just
selection core) of $(N,w) \in IG^N$ is defined as $$\mathcal{SC}(w) = \bigcup
\big\{C(v)\ |\ v \in \Sel(w) \big\} \textrm{.}$$ \end{definition}

In an analogous way as in classical games, we have a term for games with nonempty
selection core for all selections.

\begin{definition} \cite{alparslan2008cores}
	An interval game is called \emph{strongly balanced} if every
	selection of this game has a nonempty core. The set of all strongly balanced
	games on a player set $N$ is denoted by BIG$^N$.
\end{definition}

\begin{definition} The set of \emph{interval imputations} of
$(N,w) \in IG^N$ is defined as $$\mathcal{I}(w) \coloneqq \Big\{ (I_1,I_2,\ldots,I_N)
\in \mathbb{IR}^N\ |\ \sum_{i\in N} I_i = w(N),\ I_i \succeq w(i),\ \forall i
\in N \Big\} \textrm{.}$$ \end{definition}

\begin{definition} \label{def:c} The \emph{interval core} of $(N,w) \in IG^N$ is
defined as $$\mathcal{C}(w) \coloneqq \Big\{ (I_1,I_2,\ldots,I_N) \in \mathcal{I}(w)\
|\ \sum_{i \in S}I_i \succeq w(S),\ \forall S \in 2^N \setminus \{\emptyset\}
\Big\} \textrm{.}$$ \end{definition}

An important difference between the definitions of interval and selection core and imputation is that selection concepts
yield payoff vectors from $\mathbb{R}^N$, while $\mathcal{I}$ and
$\mathcal{C}$ yield vectors from $\mathbb{IR}^N$. Thus they both possess a different
degree of uncertainty.

\paragraph{Classes of interval games.}

\begin{definition} \emph{(Size monotonic interval game)} A game $(N,w) \in IG^N$ is size monotonic if for
every $T \subseteq S \subseteq N$ we have
$$|w|(T) \le |w|(S)\textrm{.}$$
That is, its length game is monotonic.
The class of size monotonic games on a player set $N$ is denoted by $\mathrm{SMIG}^N$.
\end{definition}

As we can see, size monotonic games capture situations in which an interval
uncertainty grows with the size of a coalition.

We should be careful with the following analogy of a convex game since unlike for classical games, supermodularity
is not the same as convexity.

\begin{definition} \emph{(Supermodular interval game)} An interval game $(N,w)$ is supermodular interval
if for every $S \subseteq T \subseteq N$ holds
$$w(T) + v(S) \preceq w(S \cup T) + w(S \cap T)\textrm{.}$$
\end{definition}

We get immediately that an interval game is supermodular interval if and only if its
border games are convex.

\begin{definition} \emph{(Convex interval game)} An interval game $(N,w)$ is convex interval if its
border games and length game are convex.
We write $\mathrm{CIG}^N$ for a set of convex interval games on a player set $N$.
\end{definition}

A convex interval game is supermodular as well but the converse does not
hold in general.  See \cite{alparslan2009convex} for characterizations of convex interval games
and discussion of their properties.

Finally, we define selection based classes of interval games. The paper \cite{BH15} discusses their properties and relations with the previous classes.

\begin{definition} \emph{(Selection monotonic interval game)} An interval game
$(N,v)$ is selection monotonic if all its selections are monotonic games.
The class of such games on a player set $N$ is denoted by $\mathrm{SeMIG}^N$.
\end{definition}

\begin{definition} \emph{(Selection convex interval game)} An interval game $(N,v)$
is selection convex if all its selections are convex games. The class of such games on a player set $N$ is denoted
by $\mathrm{SeCIG}^N$.
\end{definition}

\subsection{Notation}
We will use $\le$ relation on real vectors. For every $x,y \in \mathbb{R}^N$
we write $x \le y$ if $x_i \le y_i$ holds for every $1 \le i \le N$.

We do not use symbol $\subset$ in this paper. Instead, $\subseteq$ and
$\subsetneq$ are used for the subset and the proper subset relation, respectively, to avoid
ambiguity.

We also use $x(S)$ instead of $\sum_{i \in S} x_i$ occasionally.

Throughout the papers on cooperative interval games, notation, especially of
core and imputations, is not unified. It is, therefore, possible to encounter
different notation from ours. Also, in some papers the selection core is called the core of interval game.
We consider that confusing and that is why we use the term selection core
instead. The term selection imputation is then used because of its connection with
the selection core.

\section{Convexity}
\label{sec:con}

We present a characterization of the interval games in the class $\mathrm{SeCIG}$,
analogous to a classical result of Shapley on convex games \cite{shapley1971cores} and to Theorem 3.1 on convex interval games in \cite{alparslan2009convex}.

\begin{theorem} \label{thm:shlike}
For every interval game $(N,w)$, the following assertions are equivalent.
	\begin{enumerate}
		\item The game $(N,w)$ is a selection convex interval game.
		\item For every nonempty $S,T \in 2^N$, such that $S \cap T \neq T$, and $S \cap T \neq S$,
		$$\ub{w}(S) + \ub{w}(T) \leq \lb{w}(S \cup T) + \lb{w}(S \cap T).$$
		\item For every coalition $U_1, U_2, U \in 2^N$, such that $U_1 \subsetneq U_2 \subseteq N \setminus U$, and $U$ is nonempty,
		$$\ub{w}(U_1 \cup U) - \lb{w}(U_1) \leq \lb{w}(U_2 \cup U) - \ub{w}(U_2).$$
		\item For every coalition $T_1, T_2 \in 2^N$, and for every $i \in N$, 
		such that $T_1 \subsetneq T_2 \subseteq N \setminus \{ i \}$,
		$$\ub{w}(T_1 \cup \{ i\}) - \lb{w}(T_1) \leq 
		\lb{w}(T_2 \cup \{ i\}) - \ub{w}(T_2).$$
	\end{enumerate}
\end{theorem}
\begin{proof}\ \vspace{0.4mm}

$(1) \leftrightarrow (2):$ This proof is very similar to the proof
of Theorem 2 in \cite{1410.3877}.

$(2) \rightarrow (3):$ Suppose for a contradiction that there exist $U_1, U_2, U \in 2^N$, $U$ nonempty, such
that $U_1 \subsetneq U_2 \subseteq N \setminus U$, and
$$\ub{w}(U_1 \cup U) - \lb{w}(U_1) > \lb{w}(U_2 \cup U) - \ub{w}(U_2).$$

Define $S \coloneqq U_1 \cup U$, and $T \coloneqq U_2$. Both $S$ and $T$ are nonempty
sets and they are incomparable. Furthermore:
\begin{align*}
\ub{w}(U_1 \cup U) - \lb w(U_1) &> \lb w (U_2 \cup U) - \ub w (U_2), \\
\ub{w}(S) - \lb w(U_1) &> \lb w (T \cup U) - \ub w (T), \\
\ub{w}(S) + \ub w(T) &> \lb w (T \cup U) + \lb w (U_1), \\
 \ub{w}(S) + \ub w(T)&> \lb w (S \cup T) + \lb w (S \cap T).
\end{align*}
And we obtained a contradiction.

$(3) \rightarrow (4):$ Straightforward; take $U_1 \coloneqq T_1,
U_2 \coloneqq T_2, U \coloneqq \{ i \}$.

$(4) \rightarrow (3):$ Suppose that (4) holds and (3) does not. 
Take $U$ that violates (3) of minimal cardinality. If $|U| = 1$, we get
a contradiction. If $|U| > 1$, we can construct $U'$, with $|U'| = |U| - 1$,
such that it violates (3) as well. This contradicts the minimality of $U$.

$(3) \rightarrow (2):$
For a contradiction, take $S$ and $T$ which violate (2).
Define $U \coloneqq S \setminus T$; this must be nonempty since S and T are nonempty and incomparable.
Define $U_1 \coloneqq S \cap T$ and $U_2 \coloneqq T$.
As for the conditions on $U_1$ and $U_2$, we see that
$U_1 \subsetneq U_2$, since $U$ is nonempty and $(S \cap T) \subseteq T$. Now:
\begin{align*}
\ub{w}(S \cup U) - \ub w(T) &> \lb w (S \cup T) + \lb w(S \cap T) \\
\ub{w}(U_1 \cup U) + \ub w(T) &> \lb w (U_1 \cup U \cup T) + \lb w (U_1) \\
\ub{w}(U_1 \cup U) - \lb w(U_1) &> \lb w (U_2 \cup U) - \ub w (U_2) \\
\end{align*}
A contradiction.
\qed
\end{proof}

\section{Core coincidence}
\label{sec:coincidence}

In Alparslan-G\"{o}k's PhD thesis \cite{gokphd} and in paper \cite{alparslan2011set}, the following question is suggested:
\begin{quote}
\emph{``A
difficult topic might be to analyze under which conditions the set of payoff
vectors generated by the interval core of a cooperative interval game
coincides with the core of the game in terms of selections of the interval
game.''}
\end{quote}

The main thing to notice is that while the interval core gives us a set of
interval vectors, selection core gives us a set of real numbered vectors. To
be able to compare them, we need to assign to a set of interval vectors a set
of real vectors generated by these interval vectors. That is exactly what the
following function $\gen$ does.

\begin{definition} The function $\gen : 2^{\mathbb{IR}^N} \to 2^{\mathbb{R}^N}$ maps to every set of interval vectors a set of its selections. It is defined as
$$\gen(S) = \bigcup_{s \in S} \big\{(x_1,x_2,\ldots,x_n)\ |\ x_i \in s_i\big\}\textrm{.}$$
\end{definition}

The core coincidence problem can be formulated in the following way. 

\begin{problem}{(Core coincidence problem)}
What are the necessary and sufficient conditions so that an interval game satisfies
$\gen(\mathcal{C}(w)) = \mathcal{SC}(w)$?
\end{problem}

To avoid a cumbersome notation we define the following property.

\begin{definition}
	Let $(N,w)$ be a cooperative interval game. We call the game \emph{core-coincident}
	if $\gen{(\mathcal{C}(w))} = \mathcal{SC}(w)$. Also, we say that a set of interval
	games is \emph{core-coincident} if all games in this set are core-coincident.
\end{definition}

Our results in this section are an important step towards a complete classification of core-coincident
games.

\subsection{Positive results}

\begin{proposition} \label{prop:empty}
	Every cooperative interval game with empty selection core is core-coincident.
\end{proposition}
\begin{proof}
	This easily follows from \cite[Theorem 7]{BH15}. \qed
\end{proof}

\begin{proposition}
	Every degenerated cooperative interval game is core-coincident.
\end{proposition}
\begin{proof}
	It is easy to check that definitions of selection core (Definition \ref{def:sc})
	and interval core (Definition \ref{def:c}) coincide for degenerate games.
	\qed
\end{proof}

We present the following example, showing there exist infinitely many core-
coincident non-degenerated games with a nonempty
player set. But first, we need one more result.

\begin{theorem} \label{thm:technical} \emph{(Core coincidence technical lemma, \cite{1410.3877})}
For every interval game $(N,w)$ we have $\gen(\mathcal{C}(w)) = \mathcal{SC}(w)$, if and only if
for every $x \in \mathcal{SC}(w)$, there exist nonnegative vectors $l^{(x)}$ and $u^{(x)}$, such that
\begin{alignat}{3}
  &x(N) - l^{(x)}(N) &&= \underline{w}(N)\textrm{,}\\ 
  &x(N) + u^{(x)}(N) &&= \overline{w}(N)\textrm{,}\\
  &x(S) - l^{(x)}(S) &&\ge \underline{w}(S),\ \forall S \in 2^N,\\ 
  &x(S) + u^{(x)}(S) &&\ge \overline{w}(S),\ \forall S \in 2^N.
\end{alignat}
\end{theorem}

\begin{theorem}
	There are infinitely many non-degenerated core-coincident interval games.
\end{theorem}
\begin{proof}
	Define a game $(N,w_A)$, $w_A(S) \coloneqq 1/|S|$, if $S \neq N$, and further
	$w_A(N) \coloneqq [|N|,|N| + b], b > 0, b \in \mathbb R$. 

	Clearly, $\mathcal{C}(w_A)$ consists exactly of vectors
	$x$, such that $x(N) \in w_A(N)$, and $x_i \geq 1, \forall i \in N$.

	Take any such vector $x$. Define $l^{(x)}_i \coloneqq x(i) - 1$, and
	$u^{(x)}_1 \coloneqq x(1) + \ub{w_A}(N) - x(N)$, and $u^{(x)}_i \coloneqq x(i)$,
	for every $i \in N, i \neq 1$. It is now straightforward to check that all inequalities
	of Theorem \ref{thm:technical} hold and, therefore, this game is core-coincident.
	\qed	
\end{proof}

\subsection{Negative results}

\begin{theorem} \label{thm:coincidence}
	Let $(N,w)$ be an interval game such that:
	\begin{itemize}
	\item a game $(N,u)$, defined by
	$$
u(S) \coloneqq \begin{cases} \ub w (S) &\mbox{if } S = N, \\ 
\lb w (S) & \mbox{if } S \neq N, \end{cases}
	$$
	has a nonempty core, and
	\item $C(\ub w) \neq \emptyset$.
	\end{itemize}
	Then $(N,w)$ is not core-coincident.
\end{theorem}
\begin{proof}
	We define an excess function as $e(x,S) \coloneqq x(S) - w(S)$.

	If for every $x \in C(u)$, and every player $i \in N$, there is a coalition
	$S \in 2^N \setminus N, i \in S$, such that $e(x,S) = 0$, then we claim that the core of the upper border game $\ub w$ of $w$ is empty.

	To see this, observe that $C(\ub w) \subseteq C(u)$. But, if
	every $x \in C(u)$ has the aforementioned property then, by Theorem \ref{thm:technical}
	none of those $x$ can be in $C(\ub w)$; a contradiction with $C(\ub w) \neq \emptyset$.

	So the other option is that there exists a vector $y \in C(u)$, and
	a player $j \in N$, such that for every $S \in 2^N \setminus N, j \in S$, $e(x,S) > 0$.
	We define $$m \coloneqq \min_{S \in 2^N \setminus N, j \in S} e(x,S),$$
	and $M$ the set on which this minimum is attained. We pick an arbitrary player $j' \in N \setminus M$. Such player must exist.
	Then we construct a new vector $y'$:
		$$
y'_k \coloneqq \begin{cases} \lb{y_k} - m &\mbox{if } k = j,  \\ 
\lb{y_k} + m & \mbox{if } k = j', \\
\lb{y_k} & \mbox{else}. \end{cases}
	$$
	It can be checked that $y' \in C(u)$  and by a similar argument as in the previous
	case, $y'$ does not satisfy the mixed system of inequalities in Theorem \ref{thm:technical} and we are done.
	\qed
\end{proof}

This theorem has several important corollaries.

\begin{corollary}
	Every interval game $(N,w) \in \textrm{BIG}^N$ with $|w|(S) > 0$ for every $S \in 2^N$,
	is not core-coincident.
\end{corollary}

\begin{corollary}
	Classes $\textrm{SeCIG}^N$ and $\textrm{CIG}^N$ are not core-coincident for $|N| > 1$.
	Furthermore, every game in $\textrm{SeCIG}^N \cup \textrm{CIG}^N$ with every interval
	non-degenerated and $|N| > 1$ is not core-coincident.
\end{corollary}
\begin{proof}
	Theorem \ref{thm:convex} implies that selection convex games are totally balanced.
	In \cite{alparslan2009convex}, it is proved that a game is convex interval game
	if and only if its lower border game and its length game are convex. 
	This completes the proof.
	\qed
\end{proof}

Observe that $\textrm{SeCIG}^N \subseteq \textrm{SeSIG}^N \subseteq \textrm{SeMIG}^N$,
so we immediately obtain that all these sets are
not core-coincident as well for nontrivial player sets. Also,  $\textrm{CIG}^N \subseteq \textrm{SIG}^N$, so superadditive interval games are not core-coincident either.

From this we conclude that selection core and interval core behave differently on many important and widely used classes
with nontrivial uncertainty. Therefore, to further develop theory and solve problems
regarding both versions of cores of interval games is an important task.

\iffalse
\begin{theorem}
	Let $(N,w)$, be an interval game with nonempty strong core, $|N| > 1$, and $\mathcal{SC}(w) \neq \mathcal{SCC}(w)$. Then $(N,w)$ is not core-coincident.
\end{theorem}
\fi

\section{The Shapley value}
\label{sec:shapley}

\paragraph{Preliminaries and definitions.} Before we list the axioms we need in this section, a few definitions are needed.

Every function $f: IG^N \to \mathbb{IR}^N$ is called \emph{interval value function}. We
omit \emph{interval} when context is clear.

\begin{itemize}
	\item Two intervals $I, J$ are said to be \emph{indifferent}
	if $\frac{\lb I + \ub I}{2} = \frac{\lb J + \ub J}{2}$. We denote it by $I \sim J$.

	\item Let $(N,w) \in IG^N$ and $i \in N$. Then, $i$ is called a \emph{null player}
	in $w$ if $w(S) = w(S \cup i)$ for every $S \subseteq N \setminus \{i \}$.

	\item Let $(N,w) \in IG^N$ and $i \in N$. Then, $i$ is
	called a \emph{total null player} in $w$ if $w(S) \ominus w(S \cup i) = [0,0]$
	for every $S \subseteq N \setminus \{ i \}$. In other words, $i$ is 
	total null player if it is a dummy player in every selection.

	\item Let $(N,w) \in IG^N$ and $i,j \in N$. Then, $i$ and $j$ are \emph{symmetric}
	players in $(N,w)$, if $w(S \cup i) = w(S \cup j)$ for every $S \subseteq N \{ i,j \}$.
\end{itemize}

We can state a few axioms.

\begin{enumerate}
	\item Indifference efficiency (IEFF): $\sum_{i \in N} \Psi_i(w) \sim w(N)$
	for all $(N,w) \in IG^N$.

	\item Efficiency (EFF): $\sum_{i \in N} \Psi_i(w) = w(N)$
	for all $(N,w) \in IG^N$.

	\item Indifference null player property (INP): There exists a unique
	$t \ge 0$ such that $\Psi_i(w) = [-t,t]$ for any $(N,w) \in IG^N$
	and all null players $i$ in $(N,w)$.
	
	\item Total null player property (TNP): For every total null player
	$i$ in a game $(N,w)$, $\Psi_i(w)=[0,0]$.

	\item Symmetry (SYM): $\Psi_i(w) = \Psi_j(w)$ for all $(N,w) \in IG^N$
	and all symmetric players $i$ and $j$ in $(N,w)$.

	\item Additivity (ADD): $\Psi(v+w) = \Psi(v) + \Psi(w)$ for all
	$(N,v),(N,w) \in IG^N$ with $(N,v+w) \in IG^N$.
\end{enumerate}

\begin{definition} \label{def:ext} The \emph{interval Shapley value extension} is
a value function $\Phi^*: IG^N \to \mathbb{IR}^N$,
$$\Phi^*_i(w) \coloneqq \sum_{S \subseteq N\setminus \{i\}} \frac{|S|!(n-|S|-1)!}{n!} (w(S \cup i) \ominus w(S)).$$
\end{definition}

\begin{theorem} \label{thm:cinani} \cite{cinani} The function $\Phi^*$ satisfies
axioms IEFF, INP, SYM, and ADD. Furthermore, it is the only function
satisfying these axioms.
\end{theorem}

We now prove an important, yet never noted and proved property of 
the interval Shapley value extension.

\begin{theorem} \label{thm:precisely}
	For every interval game $(N,w) \in IG^N$, we have
	$$\Phi^*_i(w) =  \big\{ \phi_i(v)\, |\, v \in \Sel(w) \big\}.$$
\end{theorem}
\begin{proof}
	Every resulting value from $\Phi^*_i(w)$ corresponds
	to some selection of $w$ by Definition \ref{def:ext}
	and interval arithmetics, more precisely because of Moore's subtraction (Definition \ref{def:ari}).
	The converse holds as well. \qed
\end{proof}

In other words, the interval Shapley value extension contains exactly all
possible Shapley values that can be attained when uncertainty is settled.
We find this property very important.

However, as is noted in \cite{cinani}, efficiency is not always satisfied.
Let us explain this issue. From properties of interval arithmetics, we see
that $X - X$ is not equal to $[0,0]$ in general for $X \in  \mathbb{IR}$.
In fact, for every interval $A,B \in \mathbb{IR}$, $|A+B| \geq \min\{|A|,|B|\}$.
An analogous fact holds for Moore's subtraction as well.
Since, by definition of the interval Shapley value extension, in $\sum_{i \in N} \Phi^*_i(w)$
are some intervals added and subtracted multiple times, the resulting value does
not satisfy efficiency if we first compute $\Phi^*_i(w)$ for every $i$, and only then
add them together. This is the reason why EFF is not satisfied in general.
We can first simplify $\sum_{i \in N} \Phi^*_i(w)$ and only then add it
together. Then we would get the efficiency by the same reasoning as we get efficiency for
the Shapley value in classical games.

\paragraph{Another axiomatization.} The following theorem
shows a different axiomatization of
the interval Shapley value extension than \cite{cinani}.
We show that the axiom TNP can be interchanged with the axiom INP, which is, from
our point of view, more natural.

\begin{theorem} \label{thm:axiom}
	There is a unique value function satisfying axioms
	IEFF, TNP, SYM and ADD. Furthermore, it equals $\Phi^*$.
\end{theorem}
\begin{proof}

If a value function satisfies IEFF, INP, SYM, and ADD, then it is equal
to $\Phi^*$. From its formula, we conclude that TNP is satisfied.

Now in the other direction, if a value function satisfies
IEFF, TNP, SYM and ADD we want to show that it satisfies INP as well.

Our goal is to prove that in every game $(N',w')$ with a null player $h$,
$\phi^{**}_h(w')$ is an interval symmetric around zero.

It suffices to prove that:
\begin{itemize}
	\item If $k \in \Phi^{**}_h(w')$, then $-k \in \Phi^{**}_h(w')$, and
	\item if $a < b$, and $a,b \in \Phi^{**}_h(w')$, then $[a,b] \subseteq \Phi^{**}_h(w')$.
\end{itemize}
Both of these claims can be proved by using ADD axiom and the fact, that on
degenerated game, $\Phi^{**}$ coincides with $\phi$. We omit technical details here.

We know that every null player gets a symmetrical interval under a
value function $\phi^{**}$ satisfying IEFF, TNP, SYM and ADD. So the only remaining
option is that there must exist a game in which two null players get a different symmetrical
interval. Let us denote such game as $(N'',w'')$ and the two null players as $i$ and $j$.

Observe from the definition of null player that
$$w''(S \cup i) = w''(S)$$
holds for every $S \subseteq N \setminus \{ i \}$, and thus, specially, for every
$S \subseteq N \setminus \{ i,j \}$.
Following the same reasoning, we arrive on conclusion that
$$w''(S \cup j) = w''(S)$$
holds for every $S \subseteq N \setminus \{ i,j \}$.
Combining this, we get that
$$w''(S \cup j) = w''(S) = w''(S \cup i), \forall S \subseteq N \setminus \{ i,j \}.$$
That means that $i$ and $j$ are symmetrical and from the axiom SYM, $\phi^{**}_i(w'')$
should be equal to $\phi^{**}_j(w'')$, a contradiction.

Finally, we note that the independence of properties IEFF, TNP, SYM, and ADD follows
from Theorem \ref{thm:cinani} and from \cite{peters2015game}.
\qed
\end{proof}

\begin{theorem}
	For every $(N,w) \in \mathrm{SeCIG}$, we have $\gen(\Phi^*(w)) \subseteq \mathcal{SC}(w)$.
\end{theorem}
\begin{proof}
	From Theorem \ref{thm:precisely}, the Shapley value of every selection is in 
	$\gen(\Phi^*(w))$. Since every selection of $(N,w)$ is a convex game, its Shapley value
	lies in its core and thus also in $\mathcal{SC}(w).$ \qed
\end{proof}

\paragraph{On the improved interval Shapley-like value.} In Han et al. \cite{cinani},
an improved Shapley-like value satisfying EFF is presented.

\begin{definition} \emph{(The improved interval Shapley-like value)}
For any $(N,w) \in IG^N$ with $\sum_{i \in N} \Phi_i^*(w) \neq w(N)$,
the improved interval Shapley like value $I\Phi^*(w)$ is defined by
$$I\Phi_i^*(w) \coloneqq \Phi^*_i(w_m) + \frac{|\Phi^*_i(w_u)|}{\sum_{i \in N} |\Phi^*_i(w_u)|} \Big[ -\frac{1}{2}|w(N)|,\frac{1}{2}|v(N)| \Big] ,$$
where $w_m(S) \coloneqq (\ub w(S) + \lb w(S))/2.$
\end{definition}

\begin{theorem} \label{thm:imp}
	For every interval game $(N,w) \in IG^N$, we have
	$$I\Phi^*_i(w) \neq  \big\{ \phi_i(v)\, |\, v \in \Sel(w) \big\}.$$
\end{theorem}
\begin{proof}
	We use Theorem \ref{thm:precisely} and Definition \ref{def:ext}. \qed
\end{proof}

We believe that this is a big downside of the improved interval Shapley-like value.
We borrow a game from \cite{cinani} to illustrate the theorem.

\begin{example}
	Let $(N,v)$ be a three-person interval game where $v(1) = [0,2], v(2) = [1/2, 3/2],
	v(3) = [1,2], v(1,2) = [2,3], v(2,3) = [4,4], v(1,3) = [3,4]$, and
	$v(1,2,3) = [6,7]$. Then $\Phi^*(v) = ([11/12,31/12],[7/6,17/6],[23/12,43/12])$.
	However, $I\Phi^*(v) = ([19/12,23/12],[11/6,13/6],[31/12,35/12])$.

	By Theorem \ref{thm:imp}, there must be a selection $v'$ of $(N,v)$, such
	that $\phi_1(v') = 11/12$. But this value is not contained in $I\Phi_1^*(v)$.
\end{example}

\section{Conclusion and future research}
\label{sec:conclusion}

We investigated convexity in interval games, core coincidence problem and
interval Shapley value. To this end, we would like to summarize our results.

\begin{itemize}
	\item We showed a Shapley-like characterization of selection convex interval games in
	Theorem \ref{thm:shlike}.
	\item We tried to characterize all core-coincident games. Our main contribution is
	Theorem \ref{thm:coincidence}, saying that a large class of interval
	games is not core-coincident.
	This result implies that many classes, including CIG$^N$,  SeCIG$^N$, and strongly
	balanced games are not core-coincident.
	\item We analyzed interval Shapley value extension for interval games. We emphasized
	several facts which speak in favor of using this solution concept. Also, we showed
	a different, from our point of view more natural axiomatization of this value
	function in Theorem \ref{thm:axiom}.
\end{itemize}

Apart from the open problems presented in the papers from the references we
think it could be interesting to define prekernel for interval
games and axiomatically characterize it, analogously to Peleg
\cite{peleg1986reduced}. Also, interval games with communication structures
were not studied yet. See Bilbao's book \cite{bilbao2000cooperative} for a
theoretical background.

\section*{Acknowledgments}
The author would like to acknowledge the support of GAČR P403-18-04735S of the Czech Science Foundation and the support of the grant SVV–2017–260452.

\bibliographystyle{plain}
\bibliography{bibliography}

\end{document}